\documentclass[11pt,a4paper]{article}

\usepackage[T1]{fontenc}
\usepackage[utf8]{inputenc}
\usepackage{lmodern}
\usepackage{natbib}

\usepackage{amsmath}
\usepackage{amssymb}
\usepackage{amsthm}
\usepackage{bm}

\usepackage{graphicx}
\usepackage{booktabs}
\usepackage{array}
\usepackage{multirow}
\usepackage{float}
\usepackage{pdflscape}

\usepackage{enumitem}

\usepackage[
colorlinks=true,
linkcolor=blue,
citecolor=red,
urlcolor=magenta
]{hyperref}

\usepackage[nameinlink,noabbrev]{cleveref}

\usepackage[a4paper,
left=3cm,
right=3cm,
top=3cm,
bottom=3cm]{geometry}

\usepackage{setspace}
\usepackage{microtype}

\theoremstyle{plain}

\newtheorem{theorem}{Theorem}[section]
\newtheorem{lemma}[theorem]{Lemma}

\newtheorem{corollary}[theorem]{Corollary}

\theoremstyle{definition}

\newtheorem{definition}[theorem]{Definition}

\theoremstyle{remark}

\newtheorem{remark}[theorem]{Remark}

\DeclareMathOperator{\MAD}{MAD}
\DeclareMathOperator{\MS}{MS}
\DeclareMathOperator{\IS}{IS}
\DeclareMathOperator{\AD}{AD}

\title{\textbf{An Interval--Score ROC Curve \\for Assessment, Calibration and Ensembling \\ of Probabilistic Forecasts}}

\author{
Simone Milanesi\thanks{Department of Electrical, Computer and Biomedical Engineering, University of Pavia, Pavia, Italy}
\and
Marco Capelletti$^*$
\and
Flavio Bobba$^*$
\and
Giuseppe De Nicolao$^*$
}

\date{}

\begin{document}

\maketitle

\begin{abstract}
Probabilistic forecast evaluation is inherently multi-objective, yet existing proper scoring rules reduce predictive performance to a single scalar value, potentially obscuring the trade-off between forecast concentration and predictive accuracy. We introduce the Interval-Score Receiver Operating Characteristic (IS–-ROC) Curve, a graphical framework that represents the complete family of interval forecasts generated by varying prediction tightness. We show that the IS–-ROC Curve induced by the data generating process is Pareto optimal and convex, providing a geometric characterization of the optimal forecasting frontier. Building on these properties, we propose a geometry-based calibration procedure based on tangent optimization and convexification, together with an ensemble strategy that combines competing forecasters through convex hull construction. Finally, we provide a practical workflow and numerical examples illustrating forecast comparison, calibration, and ensemble construction within the proposed framework.

\end{abstract}

\tableofcontents

\bigskip



\section{Introduction}
\label{sec:introduction}

\subsection{Probabilistic prediction in the forecasting landscape}
\label{subsec:probabilistic_forecasting}

Forecasting was historically developed as a point prediction problem, since the objective was to estimate a single future value rather than a range of plausible values together with their associated probabilities. During the last decades, however, the focus has progressively shifted toward probabilistic forecasting. In many application domains, including epidemiology, energy systems, and finance, explicit quantification of predictive uncertainty has become the de facto standard \cite{bracher2021evaluating,zhang2014review,adrian2019vulnerable}.

Among the available approaches, probabilistic forecasting represents only one possible framework for uncertainty quantification. As discussed by Zhang et al.~\cite{zhang2014review}, point forecasts can be extended in several directions, including risk-index forecasts, which provide a static extension of point predictions, and space--time scenario forecasts, which focus on stochastic trajectories.
Although these alternative approaches offer important advantages, such as reduced computational complexity or explicit modeling of space--time dependence, they do not fully characterize the underlying uncertainty. Probabilistic forecasting is conceptually more ambitious, as it aims to approximate the entire conditional distribution of the data generating process.

As emphasized by Dawid and later stressed by Gneiting et al.~\cite{dawid1984present,gneiting2007probabilistic},
``Forecasts should be probabilistic in nature''.
Within this framework, nature generates an observation
\[
Y \sim G,
\]
where \(G\) denotes the unknown distribution of the data generating process (DGP). Forecasters have access only to partial information, that is past values of the observations and possibily of (some) covariates. Based on these, they produce a predictive distribution $F$, e.g., through quantile regression.
A forecaster that has full knowledge of the DGP $G$ is called \textit{oracle}.

The natural question then becomes:

\begin{center}
\emph{How should different probabilistic forecasts be compared?}
\end{center}


\subsection{Evaluation of point forecasters}
\label{subsec:point_forecasting}

For point forecasting, a large number of well-established evaluation tools are available.

Classical scalar error measures include

\begin{itemize}
    \item Mean Absolute Error (MAE),
    \item Mean Squared Error (MSE),
    \item Mean Absolute Percentage Error (MAPE).
\end{itemize}


Graphical tools also provide immediate visual comparisons between competing forecasters, for example through observed-versus-predicted scatter plots.
Consequently, point forecasting benefits from both scalar performance measures and intuitive graphical diagnostics.


\subsection{Current evaluation tools for probabilistic forecasting}
\label{subsec:limitations}

Several metrics have been proposed for evaluating probabilistic forecasts. Each of the currently available approaches summarizes the predictive distribution into a comparatively low-dimensional diagnostic.
A metric is said to be \textit{proper} if no model can outperform the \emph{oracle}, i.e. the predictor based on the distribution $G$ of the DGP.

Among most widely used metrics are:

\begin{itemize}

\item \textbf{Probability Integral Transform (PIT) histograms.} Flatness of the PIT histogram provides a necessary but not sufficient condition for calibration
\cite{gneiting2007probabilistic}: calibration alone cannot distinguish between genuinely informative forecasts and calibrated but underperforming predictive distributions.

\item \textbf{Interval Score (IS).} The Interval Score is a proper scoring rule for interval forecasts and evaluates predictive performance at a nominal coverage level \(1-\alpha\) \cite{gneiting2007strict}. Although extremely useful for local assessment, it provides only a partial description of predictive performance across the entire predictive distribution.

\item \textbf{Weighted Interval Score (WIS) and Continuous Ranked Probability Score (CRPS).} Both WIS and CRPS are strictly proper scoring rules and represent state-of-the-art tools for ranking probabilistic forecasts
\cite{gneiting2007strict,tibshirani2023forecast}.
Their main limitation is that they reduce the entire predictive distribution to a single scalar value. While this aggregation is highly effective for ranking competing models, it inevitably compresses diagnostic information that could otherwise reveal where and why a forecast fails (e.g., in the distribution tails rather than around its center).

\end{itemize}


\subsection{The IS--ROC curve: beyond scalar evaluation}
\label{subsec:isroc}
The comparison of probabilistic forecasting models is intrinsically a multidimensional problem. Indeed, predictive performance cannot be fully characterized by a single numerical score, since different models may exhibit different trade-offs between forecast concentration and predictive accuracy.
Traditional scoring rules impose a total ordering by collapsing these competing aspects into a scalar quantity. While convenient, such a reduction may conceal meaningful differences between competing forecasters.

As a natural extension of the Interval Score, in this work we propose a two-dimensional graphical representation of the trade-off between interval concentration and predictive error.
%
%
Models are hence compared through the set of performance points generated by varying the prediction concentration, and the most informative models are those lying on the Pareto frontier.

Within this framework, a tunable probabilistic forecasting model is viewed as a family of interval forecasters parametrized by a tightness parameter that controls the width of the interval forecast. Each value of tightness identifies a different operating point, while the collection of the associated performances summarizes the overall predictive capability of the model without requiring the user to specify \emph{a priori} a preferred trade-off between interval width and predictive accuracy.

This viewpoint is strongly inspired by the Receiver Operating Characteristic (ROC) curve used in binary classification
\cite{fawcett2006introduction,gneiting2022receiver}.
The ROC curve was originally introduced to overcome the limitations of an assessment relying on a single classification accuracy figure. Instead, the ROC curve considers an entire family of classifiers obtained by varying a decision threshold, thereby providing a much richer description of classifier performance in terms of achievable tradeoffs between true and false positive rates.

In this paper, we introduce the Interval--Score Receiver Operating Characteristic (IS--ROC) Curve, that follows the same philosophy. Instead of evaluating an interval forecast at a single tightness, it represents the entire family of interval forecasts obtained by varying the tightness.

The resulting framework enables the extension to probabilistic forecasting of several fundamental concepts from classification ROC analysis, including graphical dominance, Pareto optimality, calibration, and ensembling via convexification.




\subsection{Comparison with related work}
\label{subsec:related}

Several contributions have recognized that evaluating probabilistic forecasts requires more than a single scalar score.
One of the pioneering works is that of Christoffersen
\cite{christoffersen1998evaluating}, who introduced a likelihood-based framework for assessing interval forecasts.

Given a prediction interval for time $t$
\[
[L_t(p),U_t(p)],
\]
with nominal coverage probability \(p\), Christoffersen proposed representing forecast performance through the binary indicator sequence

\[
I_t=
\begin{cases}
0, & Y_t\in[L_t(p),U_t(p)],\\
1, & \text{otherwise}.
\end{cases}
\]
The central idea is to determine whether this sequence behaves as expected under an ideal forecaster. More precisely, \textit{conditional calibration} is decomposed into two complementary properties: \textit{unconditional calibration} (the empirical violation frequency should equal the nominal violation probability) and \textit{independence} (forecast violations should occur independently over time, rather than exhibiting temporal clustering).
Although mathematically rigorous, this framework inevitably compresses the information contained in the predictive distribution into a binary sequence, ignores the magnitude of forecast errors and evaluates only one coverage level at a time.

A different perspective is proposed by Askanazi et al.
\cite{askanazi2018comparison}, who compare interval forecasts having fixed coverage, allowing asymmetric prediction intervals.

Our framework addresses a complementary problem. Rather than varying the allocation of probability mass between the lower and upper tails for a fixed coverage, we study the entire family of central prediction intervals across all coverage levels, thereby providing a global geometric description of the trade-off between interval width and predictive error.

Another important contribution is the calibration paradigm proposed by Gneiting et al.
\cite{gneiting2007probabilistic}, summarized by the principle \char96\char96maximize sharpness subject to calibration''.
Our approach is fully compatible with this philosophy but reverses its operational perspective.
Instead of imposing calibration before model comparison, we first compare IS--ROC curves through their geometry and only afterwards apply a calibration procedure based on the geometry of the selected curve. Calibration therefore becomes a corrective transformation that is always available rather than a prerequisite preventing comparison.

\bigskip

The aggregation of predictive distributions has also received considerable attention.
Gneiting and Ranjan
\cite{gneiting2013combining}
showed that linear pooling of calibrated predictive distributions generally produces overdispersed and therefore uncalibrated forecasts. To overcome this limitation, they proposed nonlinear aggregation procedures such as the Spread-Adjusted Linear Pool (SLP) and the Beta-Transformed Linear Pool (BLP).

Our framework adopts a substantially different viewpoint. 
Instead of combining predictive distributions globally, convexification operates locally in the IS--ROC plane. 





\section{The IS--ROC Curve}
\label{sec:isroc}

\subsection{Problem framework}
\label{subsec:problem}

The theoretical framework adopted throughout this work follows that proposed by
Gneiting et al.~\cite{gneiting2007probabilistic}. 
For each time instant $t=1,2,\ldots,n$, Nature generates the observation
\[
Y_t \sim G_t(\cdot),
\]
where  $G_t(\cdot)$
denotes the (unknown) probability distribution governing the DGP.
A probabilistic forecasting model provides an estimate of this distribution, denoted by
\[
F_t(y\mid x_t)
=
\Pr(Y_t\le y\mid X_t=x_t).
\]
where \(X_t\) represents a random vector of covariates that is available to the forecaster.

For example, in wind power forecasting, the response variable \(Y_t\) may represent the generated electrical power, whereas the covariates $X_t$ may include meteorological forecasts and sensor data as well as day-of-year and time-of-day variables.

The objective of probabilistic forecasting is therefore to estimate a conditional distribution that is as close as possible to $G_t$, the distribution of the true data generating process. 
When no covariates are available, this problem reduces to estimating the marginal distribution of \(Y\), which can often be visualized through classical descriptive tools such as histograms or kernel density estimates.

The presence of covariates fundamentally changes the problem.
Rather than estimating a single probability distribution, i.e. a function of a single variable $y$, one must estimate a conditional distribution, i.e. a function of $y$ and  several covariates. Consequently, traditional graphical diagnostics may become inadequate.

Throughout this paper, we do not address the problem of learning the predictive distribution from a training set, but rather we assume that one or more probabilistic forecasting models capable of producing such conditional distributions are available and we address the problem of assessing their performance, improve their calibration and build ensembles, exploiting the availability of the dataset $\{(x_i,y_i)\}_{i=1}^{n}$.


\subsection{Mathematical definition of the IS--ROC Curve}
\label{subsec:definition}

We begin by introducing the notion of \textit{Tunable Interval Predictor}.

\begin{definition}[Tunable Interval Predictor]
Let \(x\) be a vector taking values in a measurable subset 
\( \mathcal{X} \subseteq \mathbb{R}^d\).
A \emph{Tunable Interval Predictor} (TIP)  is a mapping
\[
P :
\mathcal{X}\times [0,1]
\longrightarrow
\{\text{real intervals}\},
\]
which, for a given $x\in\mathcal{X}$ and tightness parameter
\(\beta\in[0,1]\),
yields predictive intervals

\[
P(\beta,x)
=
\left[
l_\beta(x),
u_\beta(x)
\right].
\]
satisfying the nesting and merging properties:
\begin{itemize}
\item 
\(
\beta_1<\beta_2
\Longrightarrow
l_{\beta_1}(x)
\le
l_{\beta_2}(x)
\le
u_{\beta_2}(x)
\le
u_{\beta_1}(x) \hspace{5pt}\forall x
\) (nesting)
\item
$\mu(x)
=
l_1(x)
=
u_1(x)
$ (merging for $\beta =1$)
\end{itemize}
\end{definition}
Above, $\mu(x)$
belongs to all prediction intervals, and it is also the point forecast yielded when tightness achieves its maximum value $\beta=1$.
For simplicity, we shall write $P_\beta(x)
:=P(\beta,x).$

The notion of TIP is contrasted with the class of static interval predictors, which provide interval predictions but do not have any tuning knob that controls the tradeoff between tightness and accuracy.

\begin{remark}
    We will deal with the covariate vector $x$ in a probabilistic framework. In particular, we refer to $x\in\mathcal{X}$ as a \textit{realization} of the random vector $X$.
\end{remark}
\begin{remark}

Every predictive distribution $F_t(y \mid x)$ naturally induces a TIP through its quantile function.
Indeed, if $q_\tau(x)
=
F^{-1}(\tau\mid x_t),
$ then the corresponding central prediction interval is given by

\[
l_\beta(x)
=
q_{\beta/2}(x),
\qquad
u_\beta(x)
=
q_{1-\beta/2}(x).
\]
Conversely, a probabilistic forecaster can equivalently be represented by the family of all its central prediction intervals.
\end{remark}

\bigskip

Associated with any interval TIP $P$ there are two fundamental quantities.
The first is its \emph{sharpness}, defined as the interval width
\[
\mathrm{S}^P(\beta,x)
=
u_\beta(x)-l_\beta(x).
\]
Averaging over the covariate distribution gives the mean sharpness
\[
\MS^{P}(\beta)
=
\mathbb{E}_X
\left[
S^P(\beta, X)
\right],
\]
which, in practice, is estimated by the sample mean
\[
\MS^{P}(\beta)
\approx
\frac1t
\sum_{i=1}^{t}
S^{P}(\beta,x_i).
\]
Sharper forecasts correspond to smaller interval widths and therefore convey more concentrated predictive information.

However, intervals that are excessively narrow inevitably increase the probability that future observations fall outside the predicted range.
To quantify this phenomenon, we introduce the \textit{Absolute Distance (AD)}:
\[
\AD^{P}(\beta,x,y)
=
\left(y-u_\beta(x)\right)_+
+
\left(l_\beta(x)-y\right)_+,
\]
where $(a)_+ = \max\{a,0\}$.
Thus, the Absolute Distance is equal to zero whenever the observation belongs to the prediction interval and increases linearly with the distance from the closest interval endpoint otherwise.
The corresponding average is the expected absolute interval distance.
\[
\MAD^{P}(\beta)
=
\mathbb{E}_{X,Y}
\left[
\AD^{P}(\beta,X,Y)
\right],
\]
which can be estimated through the sample mean
\[
\MAD^{P}(\beta)
\approx
\frac1t
\sum_{i=1}^{t}
\AD^{P}(\beta,x_i,y_i).
\]

A reduction in sharpness (i.e., narrower prediction intervals) generally increases the expected distance from future observations, whereas wider intervals reduce the prediction error at the expense of informativeness.
Consequently, these two objectives cannot usually be optimized simultaneously, giving rise to an intrinsic trade-off between concentration and predictive accuracy.

For a \emph{covariate-free} TIP, i.e. a predictor that does not uses covariates, the sharpness is deterministic, whereas the expectation defining \(\MAD^P(\beta)\) is taken only with respect to the random variable \(Y\).

The central idea of the proposed framework is to avoid collapsing this trade-off into a single scalar score.
For every value of the tightness parameter \(\beta\), the forecaster generates the performance point
\[
\left(
\MS^{P}(\beta),
\MAD^{P}(\beta)
\right)
\in\mathbb{R}^2.
\]
Allowing \(\beta\) to vary over the interval \((0,1]\) the set of all performance points draws a parametric curve, displaying all the tradeoffs achievable by the PIT.

\begin{definition}[IS--ROC Curve]
The \emph{Interval Score Receiver Operating Characteristic (IS--ROC) Curve} associated with the interval forecaster \(P_{\beta}\) is the parametric curve

\[
\beta
\longmapsto
\left(
\MS^{P}(\beta),
\MAD^{P}(\beta)
\right),
\qquad
\beta\in(0,1].
\]

\end{definition}

The IS--ROC Curve summarizes the complete operational behavior of a PIT across all prediction interval widths.



\subsubsection*{Relationship with the Interval Score}

In the proposed framework, the classical Interval Score admits a very meaningful geometric interpretation.
For an interval forecaster \(P_{\beta}\) and a given nominal miscoverage level \(\alpha\), the Interval Score can be written as

\[
\IS^{P}(\alpha,\beta)
=
\MS^{P}(\beta)
+
\frac{2}{\alpha}
\MAD^{P}(\beta).
\]
Hence, fixing \(\alpha\) defines a family of level sets
\[
\IS
=
\MS
+
\frac{2}{\alpha}\MAD
\]
or equivalently,
\[
\MAD
=
-\frac{\alpha}{2}\MS
+
\frac{\alpha}{2}\IS,
\]
where \(\frac{\alpha}{2}\IS\) is the intercept. Therefore, all points having the same Interval Score lie on parallel straight lines with slope $-\frac{\alpha}{2}$.
Under suitable regularity assumptions (for example, if the IS--ROC Curve is convex), the optimum is attained at the tangency point between the curve and a suitable level line.

Consequently, the parameter \(\alpha\) can be interpreted as the desired tradeoff between  mean sharpness and mean absolute error which, by the tangent method, leads to the selection of the optimal  tightness $\beta$.
This geometric interpretation plays a central role in the calibration procedure developed in \cref{sec:calibration}.

\subsection{Theoretical properties of the IS--ROC Curve}
\label{subsec:properties}

The IS--ROC Curve inherits a fundamental property of the Interval Score, namely its ability to associate to the oracle an optimal frontier in the $(\mathrm{S}, \mathrm{MAD})$ plane. 


\subsubsection{Pareto optimality of the oracle}

\begin{definition}[Pointwise dominance]

Let $P^{1}$ and $P^{2}$ be two TIPs.

For fixed values $\beta_{1},\beta_{2}\in(0,1]$, we say that
$P^{1}_{\beta_{1}}$
(non-strictly) dominates
$P^{2}_{\beta_{2}}$
whenever

\[
\MS^{P^{1}}(\beta_{1})
\le
\MS^{P^{2}}(\beta_{2})
\]
\[
\MAD^{P^{1}}(\beta_{1})
\le
\MAD^{P^{2}}(\beta_{2}).
\]

\end{definition}

\begin{definition}[Global dominance]

Let $P^{1}$ and $P^{2}$ be two TIPs.

We say that $P^{1}$ globally dominates $P^{2}$ if, for every
$\beta_{2}\in(0,1]$, there exists
$\beta_{1}\in(0,1]$
such that
$P^{1}_{\beta_{1}}$
pointwise dominates
$P^{2}_{\beta_{2}}$.

\end{definition}

It is immediate to see that, in the $(\MS,\MAD)$ plane, global dominance implies that the IS-ROC Curves do not cross each other, as specified by the following corollary.
\begin{corollary}
If $P^{1}$ globally dominates $P^{2}$, then if $\MS^{P^{1}}(\beta_{1})=\MS^{P^{2}}(\beta_{2})$, then $\MAD^{P^{2}}(\beta_{2}) \ge \MAD^{P^{1}}(\beta_{1})$.
\end{corollary}

\begin{definition}[Pareto optimality]

A TIP $P$ is said to be Pareto optimal if it globally dominates every other interval forecaster.

\end{definition}

\begin{theorem}
\label{lem:isproper}

Let $G$ be the TIP induced by the DGP.
Then

\[
IS^{P}(\alpha,\beta)
\ge
IS^{G}(\alpha,\alpha)
\]

for all TIPs $P$, 
$\forall \alpha\in(0,1]$,
$\forall \beta\in(0,1]$.

\end{theorem}

\begin{proof}

The result is well known in the literature. It follows from the decomposition of the Interval Score into a linear combination of pinball losses together with the strict propriety of the corresponding scoring rule.
See Gneiting and Raftery
\cite{gneiting2007strict}
for details.

\end{proof}

\begin{theorem}[Pareto optimality of the IS--ROC Curve of the oracle]
\label{thm:pareto}

Let $G$ be the TIP induced by the DGP.
Then its IS--ROC Curve is Pareto optimal.

\end{theorem}

\begin{proof}

Assume, by contradiction, that the IS--ROC Curve associated with $G$ is not Pareto optimal.
Then there exist $\alpha,\beta\in(0,1]$ and an interval forecaster $P$ such that
\[
\MS^{P}(\beta)
\le
\MS^{G}(\alpha)
\text{ and }
\MAD^{P}(\beta)
\le
\MAD^{G}(\alpha),
\]

with at least one inequality being strict.
Suppose, without loss of generality, that
$\MAD^{P}(\beta)
<
\MAD^{G}(\alpha).
$

Since

\[
\IS^{P}(\alpha,\beta)
=
\MS^{P}(\beta)
+
\frac{2}{\alpha}
\MAD^{P}(\beta),
\]
we obtain
\[
\IS^{P}(\alpha,\beta)
<
\MS^{G}(\alpha)
+
\frac{2}{\alpha}
\MAD^{G}(\alpha)
=
\IS^{G}(\alpha,\alpha),
\]
which contradicts \Cref{lem:isproper}.

\end{proof}

The previous theorem admits a simple graphical interpretation: the IS--ROC Curve associated with the true data generating process coincides with the Pareto frontier in the $(\mathrm{S},\MAD)$ plane. Consequently, no competing TIPs can generate operating points lying strictly below and to the left of the ideal curve.




\subsubsection{Convexity of the oracle}

In this subsection we show that the IS--ROC Curve associated with the DGP is always convex. We provide two proofs. The first one, restricted to the covariate-free case, relies on stronger regularity assumptions and gives an explicit analytical derivation. The second proof is more general and exploits a technical convexification lemma proved in Appendix~\ref{app:convexification}.
\begin{theorem}
\label{thm:convexity_simple} Consider the case without covariates, i.e. $G_t=G \hspace{5pt} \forall t$, and
assume that $G(y),S(\beta),
MAD(\beta)$ are strictly monotonic and differentiable.
If the IS--ROC Curve $\MAD(S)$ induced by $G$ is twice differentiable, then, it is strictly convex.
\end{theorem}
\begin{proof}
We want to prove that
\[
\frac{d^2\MAD}{d\mathrm{S}^2}>0.
\]
Since the IS--ROC Curve is defined available through the parametrization
\[
\beta
\longmapsto
\left(
\mathrm{S}(\beta),
\MAD(\beta)
\right),
\]
we first compute
\[
\frac{d\MAD}{d\mathrm{S}}
=
\frac{\MAD'(\beta)}
{\mathrm{S}'(\beta)}.
\]
The sharpness is
\[
\mathrm{S}(\beta)
=
q_{1-\beta/2}
-
q_{\beta/2},
\]
where $\gamma\mapsto q_\gamma$ denotes the quantile function of \(G\).
Using the classical identity
\[
\frac{dq(\gamma)}{d\gamma}
=
\frac{1}{f(q(\gamma))},
\]
where \(f_G\) denotes the density of \(G\), we obtain
\[
\mathrm{S}'(\beta)
=
-
\frac{1}{2f_G(B)}
-
\frac{1}{2f_G(A)},
\]
where, for brevity,
$A=q_{\beta/2} \text{ and }
B=q_{1-\beta/2}.$
Now consider
\[
\MAD(\beta)
=
\mathbb{E}
\left[
(Y-B)_+
+
(A-Y)_+
\right].
\]
Differentiating the two terms separately yields
\[
\frac{\partial}{\partial B}
\mathbb{E}(Y-B)_+
=
-
\Pr(Y>B)
=
-
\frac{\beta}{2},
\]
and
\[
\frac{\partial}{\partial A}
\mathbb{E}(A-Y)_+
=
\Pr(Y<A)
=
\frac{\beta}{2}.
\]
Therefore,
\[
\MAD'(\beta)
=
\frac{\beta}{4}
\left(
\frac1{f(B)}
+
\frac1{f(A)}
\right).
\]
Combining the previous expressions gives
\[
\frac{d\MAD}{d\mathrm{S}}
=
-
\frac{\beta}{2}.
\]
Differentiating once more,
\[
\frac{d^2\MAD}{d\mathrm{S}^2}
=
-\frac12
\frac{d\beta}{d\mathrm{S}}.
\]
Finally, due to the strict decreasing monotonicity of $\mathrm{S}(\beta)$,
we conclude that
\[
\frac{d^2\MAD}{d\mathrm{S}^2}>0,
\]
which proves the strict convexity of the IS--ROC Curve.

\end{proof}

\begin{theorem}
\label{thm:convexity_general}

Let \(G\) be an arbitrary data generating process, possibly depending on covariates.
Then the IS--ROC Curve induced by \(G\) is convex.

\end{theorem}

\begin{proof}
Let
\[
c:\MS\longmapsto \MAD
\]
denote the IS--ROC Curve associated with G.
and assume, by contradiction, that it  is not convex.
Then, on the curve, there exist two points
\[
(s_1,c(s_1)),
\qquad
(s_2,c(s_2)),
\]
such that their convex envelop lies below the curve, i.e. 
\[
tc(s_1)
+
(1-t)c(s_2)
<
c\!\left(
ts_1+(1-t)s_2
\right), \forall t\in(0,1).
\]
Define the function
\[
\widehat c(s)
=
\begin{cases}
c(s),
&
s\le s_1,
\\[1ex]
tc(s_1)
+
(1-t)c(s_2),
&
s\in[s_1,s_2],
\\[1ex]
c(s),
&
s\ge s_2,
\end{cases}
\]
where
\[
t
=
\frac{s_2-s}{s_2-s_1}.
\]
By Lemma~\ref{lem:convexification}, proved in Appendix~\ref{app:convexification}, the  segment between $s_1$ and $s_2$ can be interpreted as a randomized TIP obtained by mixing the two endpoint predictors.
Consequently, \(\widehat c\) is the IS-ROC Curve of an admissible interval forecaster.
Since the segment lies strictly below the original curve, \(\widehat c\) globally dominates the IS--ROC Curve associated with the DGP.
This contradicts the Pareto optimality established in Theorem~\ref{thm:pareto}.
Therefore, the IS--ROC Curve associated with the true data generating process must be convex.

\end{proof}

While the IS–ROC Curve associated with the oracle is always convex, this property does not generally extend to arbitrary probabilistic forecasters. Nevertheless, every IS–ROC Curve satisfies a quasi-convexity property, as proved in Appendix~\ref{app:monotonicity}.

\subsubsection{Pareto optimality is not unique to the oracle}
\label{subsec:nonuniqueness}

Although the IS--ROC Curve associated with the DGP is Pareto optimal, Pareto optimality alone does not uniquely identify the underlying predictive distribution.
Indeed, distinct predictive distributions may generate exactly the same set of trade-offs between mean sharpness and mean absolute distance. In fact, the IS--ROC Curve depends only on the set of  prediction intervals that can be generated by the forecaster by a proper selection of the tightness $\beta$. Consequently, it is well possible that distinct distributions yield the same prediction intervals, but with different choices of $\beta$.

In the covariate-free case, this feature is easily understood: as shown below, a simple example is provided by symmetric distributions sharing the same median. It is  immediate to see that these distributions share the same IS--ROC Curve.
More in general, one can find several distributional families that generate identical IS--ROC Curves. Again, there is a parallelism with the properties of the classical ROC Curve for binary classifiers, as it is well known that distinct classifiers may share the same ROC Curve.


\begin{theorem}[Non-uniqueness of the forecasters achieving the Pareto frontier]
\label{thm:nonuniqueness}
Consider the case without covariates.
Let $F$ and $G$ be two symmetric distributions such that
$\operatorname{median}(F)
=
\operatorname{median}(G)$,
where $G$ is the DGP.
Then, the IS--ROC Curve induced by $F$ coincides with the ideal IS--ROC Curve induced by $G$.

\end{theorem}

\begin{proof}

Consider the IS--ROC Curve induced by $G$.
By construction, for every actual sharpness level
$s^\star \geq 0$, there exists a value $\alpha\in(0,1]$
such that
$\mathrm{S}^G(\alpha)=s^\star$ .
Let 
$\MAD^\star
:=
\MAD^G(\alpha).$ We prove that there exists a value
$\beta\in(0,1]$
such that
$\mathrm{S}^F(\beta)=s^\star$
and
$\MAD^F(\beta)=\MAD^\star.$
Without loss of generality, and only for notational convenience, we assume that the common median of $F$ and $G$ is equal to zero.
The sharpness associated with $G$ is given by
\[
\mathrm{S}^G(\alpha)
=
u^G_\alpha-l^G_\alpha
=
q^G_{1-\alpha/2}
-
q^G_{\alpha/2},
\]
where $q^G$ denotes the quantile function of $G$.
Since $G$ is symmetric around zero, we have
$q^G_{1-\alpha/2}
=
\frac{s^\star}{2},$
and
$q^G_{\alpha/2}
=
-\frac{s^\star}{2}.$
Therefore,
$\frac{\alpha}{2}
=
1-G\left(\frac{s^\star}{2}\right),$
and, by symmetry,
$\frac{\alpha}{2}
=
G\left(-\frac{s^\star}{2}\right).$
Hence,
\[
\alpha
=
2G\left(-\frac{s^\star}{2}\right).
\]
The previous derivation relies only on the symmetry assumption. Therefore, the same argument can be applied to the distribution $F$.
Consequently, there exists $\beta\in(0,1]$
such that $\mathrm{S}^F(\beta)=s^\star$, with
$\beta
=
2F\left(-\frac{s^\star}{2}\right).
$
It remains to prove that $\MAD^F(\beta)
=
\MAD^\star .$
Using the symmetry of $F$, the two contributions outside the prediction interval are equal, and therefore
\[
\MAD^F(\beta)
=
\mathbb{E}
\left[
(Y-u^F_\beta)_+
\right]
+
\mathbb{E}
\left[
(l^F_\beta-Y)_+
\right]
=
2\mathbb{E}
\left[
(Y-u^F_\beta)_+
\right].
\]
Analogously, for $G$,
\[
\MAD^\star
=
\MAD^G(\alpha)
=
2\mathbb{E}
\left[
(Y-u^G_\alpha)_+
\right].
\]
Thus, it is sufficient to prove that the upper endpoints of the two prediction intervals coincide.
Since $\alpha =
2G\left(-\frac{s^\star}{2}\right)$,
$u^G_\alpha
=
q^G_{1-\alpha/2}
=
\frac{s^\star}{2}.
$
Similarly, from the definition of $\beta$,
$u^F_\beta
=
q^F_{1-\beta/2}
=
\frac{s^\star}{2}.$
Hence,
\[
u^F_\beta
=
u^G_\alpha .
\]
Therefore,
\[
\MAD^F(\beta)
=
\MAD^G(\alpha)
=
\MAD^\star,
\]
which proves that every point of the IS--ROC Curve generated by $G$ is also generated by $F$.
Since the argument holds for every $s^\star\geq0$,
the two IS--ROC Curves coincide.
\end{proof}

\begin{remark}

Since the right derivative of the ideal IS--ROC Curve at
$s=0$ is given by the value $-\frac{\beta}{2}$, no IS--ROC Curve can belong to the triangular region identified by the coordinate axes and the straight line passing through
$(0,\MAD^{G}(1))$
with slope $-1/2$.

\end{remark}

\section{Calibration and Convexification of IS--ROC Curves}
\label{sec:calibration}

The previous section established that the IS--ROC Curve associated with the DGP is both Pareto optimal and convex.
Let now consider what happens when the performances of several Tunable Interval Predictors (TIPs) are compared on a given data set. If we look at their IS--ROC Curves, three situations may occur (see Figure \ref{fig:img_esempi_comportamenti}):

\begin{enumerate}
    \item a dominant IS--ROC Curve exists and is convex;
    \item a dominant curve exists but is not convex;
    \item no dominant curve exists.
\end{enumerate}

This section introduces a calibration procedure capable of handling all three situations.
The basic idea is analogous to ROC analysis. First, through dominance and convexification, we get closer to the efficient frontier and then we compute a calibration map that associates each value of the Interval Score parameter $\alpha$ with the corresponding optimal tightness $\beta$.

\begin{figure}[H]
\centering
\includegraphics[height=\textwidth, angle=-90]{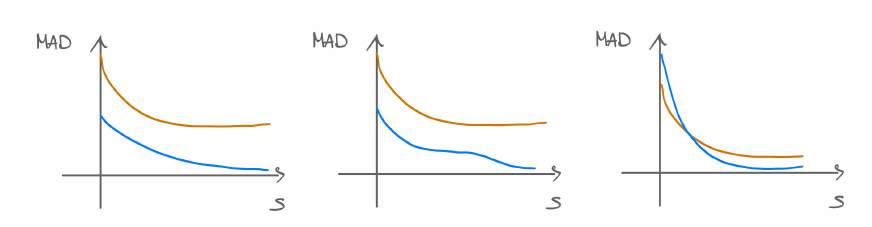}
\caption{Illustration of the three possible configurations of IS--ROC Curves. On the left, the existence of a dominant convex curve; in the middle, the existence of a dominant but non-convex curve; on the right, the absence of a dominant curve.}
\label{fig:img_esempi_comportamenti}
\end{figure}

\subsection{Calibration of a TIP}
\label{subsec:calibration}
The calibration of a TIP serves at least two purposes. First, it optimizes the interval score performances of the TIP by optimally adapting the tightness $\beta$ to the desired coverage $\alpha$. Second, when the TIP is not a probabilistic one, it converts it into a probabilistic model. Both tasks have a parallel in the calibration of tunable classifiers based on the ROC curve. For instance, by ROC calibration, a Support Vector Classifier, whose rationale is intrinsically deterministic, is enabled to yield probabilities. Calibration could be used also to obtain quantiles from an expectile-based TIP, whose tightness parameter specifies expectile thresholds rather than quantile levels.
\begin{definition} \emph{Calibration of a TIP}. The calibration of a TIP consists in computing the map $g:(0,1]\longrightarrow(0,1]$, where
\begin{equation}
\label{eq.recal}   
g(\alpha) = \arg \min_\beta \IS(\alpha,\beta)= \arg \min_\beta \left ( \MS(\beta)
+
\frac{2}{\alpha}
\MAD(\beta) \right ).
\end{equation}
\end{definition}
The interpretation is straightforward. In order to minimize $\IS(\alpha,\beta)$,
instead of using the nominal tightness level $\beta=\alpha$, one can better employ an adjusted tightness 
$\beta=g(\alpha)$.

If the TIP originates from a probability distribution, once the map  $g$ has been determined, it can be used to calibrate the predictive distribution through a transformation of the original quantiles.
More precisely, let $\tau \in (0,1)$ denote the quantile level of the original predictive distribution. The corresponding calibrated quantile level is defined as

\begin{equation} \label{eq:calibrated_tau}
  \tilde{\tau} = 
  \begin{cases}
    \dfrac{g(2\tau)}{2}, & \tau \leq \tfrac{1}{2}, \\[1.2ex]
    1 - \dfrac{g\!\left(2 - 2\tau\right)}{2}, & \tau > \tfrac{1}{2}.
  \end{cases}
\end{equation}

The calibrated predictive distribution is then obtained by evaluating the original quantile function at the transformed level $\tilde{\tau}$.

Finally, observe that calibration affects the parameterization of the IS--ROC Curve but does not affect its geometry. Consequently, two TIPs sharing the same IS--ROC Curve are equivalent up to calibration.

\subsection{Convex case: calibration by the tangent method}
\label{subsec:dominant}

Assume that TIP \(P\) generates a dominant IS--ROC Curve that is already convex.
For a given value \(\alpha\in(0,1]\) of the Interval Score parameter, the set of points in the $(\MS,\MAD)$ plane that achieve interval score $\IS$ satisfy the equation
$$
\IS = \MS + \frac{2}{\alpha} \MAD
$$
that defines a family of parallel straight lines parametrized by $\IS$. It is immediate to see that in the $(\MS,\MAD)$ plane, these lines have slope $-\alpha/2$.

Since the IS--ROC curve is convex, for every value of \(\alpha\) there exists a unique supporting line with that slope which is tangent to the IS--ROC Curve.
Then, $\beta=g(\alpha)$ defined in (\ref{eq.recal}) coincides the value of the tightness parameter corresponding to this tangency point.


















\subsection{Non convex case: convex hull and calibration}
\label{subsec:convexification}

The previous subsection assumes that the dominant IS--ROC Curve is convex.
In practice, however, a dominant forecaster may generate a non-convex set, see \ref{app:monotonicity}. In this case, the calibration procedure described above cannot be applied directly because, for some values of \(\alpha\), several tangency points may exist or no tangency point may exist at all.

A natural solution consists of replacing the IS--ROC Curve with its convex hull.
The concept is analogous to the convex hull construction commonly used in ROC analysis to combine multiple classifiers
\cite{provost1997analysis}. When different classifiers generate different efficient frontiers, the operating points are combined to obtain the upper convex envelope, which represents the best achievable performance.
Likewise, in the \((\MS,\MAD)\) plane, the efficient frontier generated by the TIP is the convex hull of its IS--ROC curve.

By Lemma~\ref{lem:convexification}, proved in Appendix~\ref{app:convexification}, every point belonging to a convexification segment admits a probabilistic interpretation.
More precisely, each point on the segment corresponds to the performance of an \emph{interpolation interval predictor} obtained by sampling the interval predictions associated with the two endpoints of the segment. The sampling ratio will control where the resulting performance point lies
Therefore, the convex hull does more than draw artificial lines in the (MS, MAD) plane; in fact, it  introduces new performance points that can be achieved through suitable interpolation strategies.

\subsubsection*{Calibration of convexified curves}
Once the convex hull has been computed, its boundary can be partitioned into subsets of two types:
\begin{itemize}
\item strictly convex curves inherited from the original IS--ROC Curve;
\item linear segments, e.g. those introduced by the convexification procedure.
\end{itemize}
While the strictly convex curves are calibrated as described in Section~\ref{subsec:dominant}, the linear segments deserve a comment. First of all note that, letting $-\hat\alpha/2$ be the slope of the segment, it represents an iso-performance set because all its points achieve the same Interval Score $\IS(\hat\alpha)=\MS + \frac{2}{\hat\alpha} \MAD$. Therefore, in this case the calibration map (\ref{eq.recal}) does not return a unique $\beta$, but an interval of values $(\beta_A,\beta_B)$ where $\beta_A$ and $\beta_B$ are the tightnesses of the endpoints of the segment. In view of this non-uniqueness, we prefer to let $g(\hat\alpha$) undefined. If we let $\alpha_A$ and $\alpha_B$ be such that
\[
\beta_A=g(\alpha_A),
\qquad
\beta_B=g(\alpha_B),
\]
the calibration over the convexified segment will be a step function:
\[
g(\alpha)
=
\begin{cases}
\beta_A,
&
\alpha\in\alpha_A,\hat{\alpha}]
\\[1ex]
\mathrm{undefined}, &\alpha = \hat \alpha
\\[1ex]
\beta_B
&
\alpha\in[\hat{\alpha},\alpha_B]
\end{cases}
\]
We refer to this feature as the
\emph{step calibration of convexified segments}. It is worth noting that, if the calibrated curve is translated in a probability distribution via quantile calculation as explained in Subsection \ref{subsec:calibration}, the convexified segments will produce constant segments in the probability distribution. As an example, in Appendix \ref{app:linear}, we will present an IS-ROC curve made by segments, whose associated distribution is a staircase function.


\subsection{Ensemble forecasting through convexification and calibration}
\label{subsec:multiple}

We now compare multiple TIPs through their (convexified)  IS--ROC Curves.
If one curve globally dominates all the others, the procedures described in Sections~\ref{subsec:dominant} and~\ref{subsec:convexification} apply directly.

Otherwise, two or more IS--ROC Curves intersect.
In this situation, a new convex hull is computed over the union of the convexified IS--ROC curves.
The resulting hull defines the global efficient frontier.

Since each convexified curve already inherits the calibration of the corresponding forecaster, only the newly created convexification segments will undergo the step calibration described in Subsection \ref{subsec:convexification}.

A major advantage of the proposed methodology is its computational efficiency. It suffices the ability to interrogate an ensemble of TIPs, without any knowledge of their internal algorithms. Indeed, for any given value of the Interval Score parameter $\alpha$, only the comparison of the IS--ROC curves is needed in order to decide which of the competing TIPs should be used to provide the interval prediction. In turn, this comparison is equivalent to selecting the TIP minimizing $IS(\alpha,g(\alpha))$.

It is desirable that an ensemble TIP produces predicted intervals that meet the nesting property defined in Subsection \ref{subsec:definition}. Unfortunately, convexification may lead to violation of this property, although an isotonic regression step may be introduced in order to restore nesting.


\section{User guide: How to construct and use the IS--ROC Curve}
\label{sec:userguide}

This section describes the practical construction of the IS--ROC Curve from a dataset and illustrates the proposed comparison procedure.

Suppose that a dataset $\{(x_i,y_i)\}_{i=1, \ldots, n}$ is available, where the pair $(x_i,y_i)$ denotes a realization of the joint random variables $(X, Y)$.
Based on the observation \(x_t\), a probabilistic forecaster produces the predictive distribution

\[
F_t(y\mid x)
=
\Pr(Y_t\le y\mid X_t=x_t).
\]

Now, let us consider a TIP, that is predictor with a tunable tightness parameter $\beta$. Fix a tightness level \(\beta\in(0,1]\).
For every observation \((x_i,y_i)\), the
predictor associates one point in the \((S,MAD)\)-plane, whose coordinates are given by
\begin{align}
\mathrm{S}_\beta^{P}(x_i)
&=
u_\beta(x_i)-l_\beta(x_i),\\
\AD_\beta^{P}(x_i,y_i)
&=
\left(
y_i-u_\beta(x_i)
\right)_+
+
\left(
l_\beta(x_i)-y_i
\right)_+.
\end{align}
Therefore, considering all observations in the dataset, a scatter plot is obtained. The point of the IS--ROC Curve corresponding to the selected value of \(\beta\) is defined as the mean of the scatter plot and its coordinates can be estimated by the sample means:
\begin{align}
\MS^{P}(\beta)
&=
\mathbb{E}_X
\left[
\mathrm{S}_\beta^{P}(X)
\right]
\approx
\frac1n
\sum_{i=1}^{n}
\mathrm{S}_\beta^{P}(x_i),
\\
\MAD^{P}(\beta)
&=
\mathbb{E}_{X,Y}
\left[
\AD_\beta^{P}(X,Y)
\right]
\approx
\frac1n
\sum_{i=1}^{n}
\AD_\beta^{P}(x_i,y_i).
\end{align}
Allowing the tightness parameter to vary over the interval $\beta\in(0,1]$ yields an estimate of the complete IS--ROC Curve.
Indeed, as the tightness $\beta$ is changed, the corresponding scatter plot moves on the \((S,MAD)\)-plane, while the trajectory of its mean point draws the IS--ROC Curve.

In the covariate-free case, the sharpness coordinate becomes deterministic because it depends only on $\beta$, whereas the Absolute Distance continues to depend on the observed responses \(Y_t\).

\subsection*{Operational procedure}

\begin{enumerate}

\item
For each  TIP \(P\), compute

\[
\left(
\MS^{P}(\beta),
\MAD^{P}(\beta)
\right),
\qquad
\beta\in(0,1],
\]

thus obtaining its IS--ROC Curve.

\item
Compare the resulting curves.

Exactly one of the following situations occurs:

\begin{enumerate}
    \item a dominant IS--ROC Curve exists and is convex;
    \item a dominant curve exists but is not convex;
    \item no dominant curve exists.
\end{enumerate}

\item
Proceed according to the identified scenario.

\begin{itemize}

\item
In Case (a), directly apply the calibration procedure [\ref{subsec:dominant}].
\item
In Case (b), compute the convex hull of the dominant curve and then apply the calibration procedure [\ref{subsec:convexification}].

\item
In Case (c), convexify each individual curve whenever necessary, construct the global convex hull, and finally apply the step calibration procedure [\ref{subsec:multiple}].

\end{itemize}

\end{enumerate}

For the computation of convex hulls, we recommend Andrew's monotone chain algorithm, which efficiently extracts the vertices of the convex envelope from the sampled points of the IS--ROC Curves.

\section{Numerical Examples}
\label{sec:results}

This section illustrates the proposed methodology through two synthetic examples.

The first example highlights the notions of dominance, non-strict dominance, and calibration. The second example considers intersecting IS--ROC Curves and illustrates the convexification and calibration procedure.

\subsection{Example 1: Three forecasters sharing the same median}
\label{subsec:example1}

We first consider the case without covariates.
The DGP is assumed to be the standard Gaussian distribution $G=\mathcal N(0,1)$, and three probabilistic forecasters are compared.

The first forecaster is the oracle, i.e., $F_1=G$.
The second forecaster is a Laplace distribution centered at the origin, i.e., $F_2=\mathrm{Laplace}(0,1)$. The third forecaster is a lognormal distribution, $F_3=\mathrm{LogNormal}(-0.2,1)$.

The purpose of this example is to illustrate three different phenomena. First, the IS--ROC Curve associated to the Gaussian forecaster dominates the lognormal one. Second, the Gaussian and Laplace forecasters generate identical IS--ROC Curves, providing an example of non-strict dominance. Finally, the calibration procedure (see Section \ref{subsec:dominant}) applied to $F_2$ allows the original predictive distribution $G$ to be reconstructed.

\begin{figure}[H]

\centering

\includegraphics[width=0.8\textwidth]{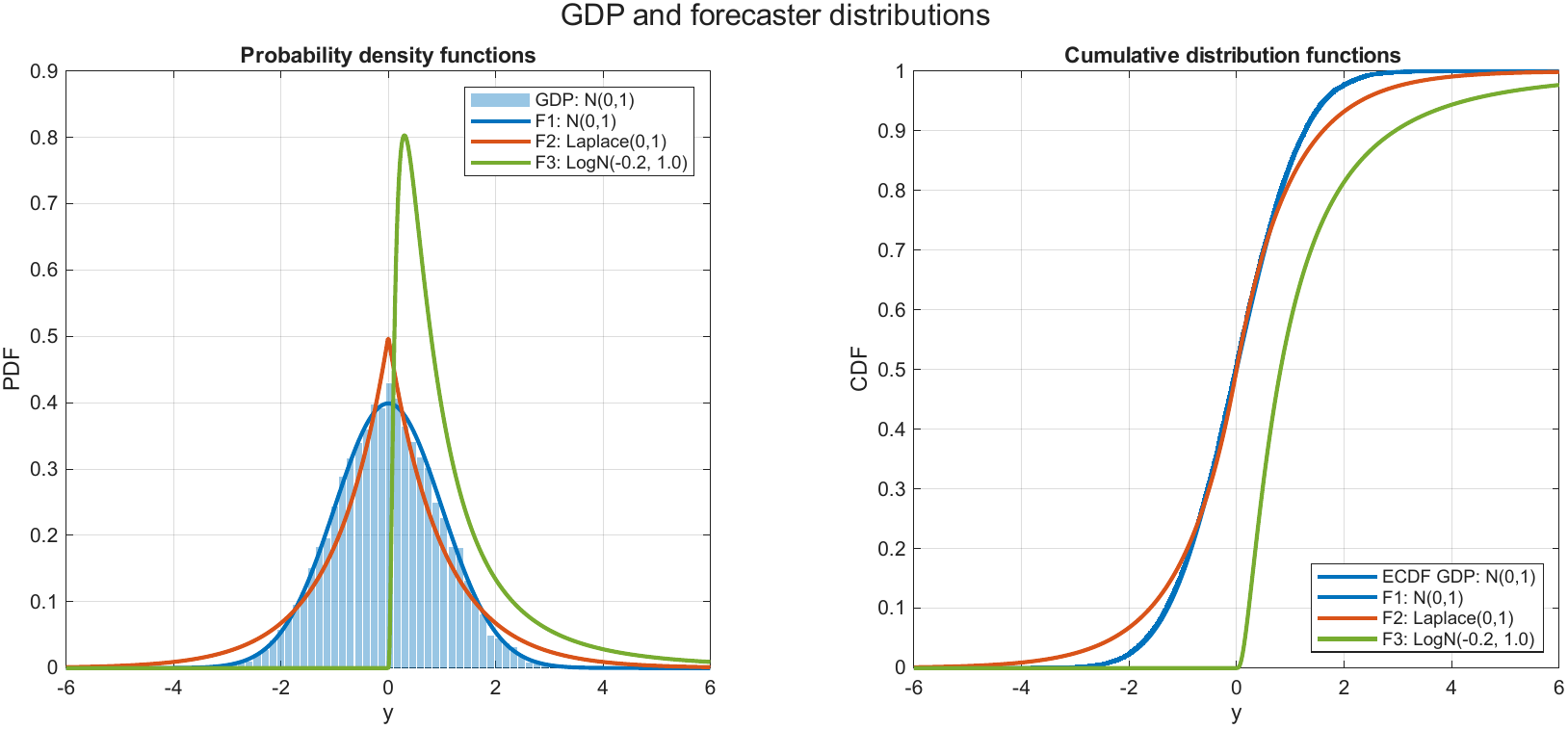}

\caption{Example 1. Comparison of the DGP and the three probabilistic forecasters.
The left panel reports the probability density functions, whereas the right panel shows the corresponding cumulative distribution functions.
}

\label{fig:example1_pdf}

\end{figure}

\begin{figure}[H]

\centering

\includegraphics[width=0.8\textwidth]{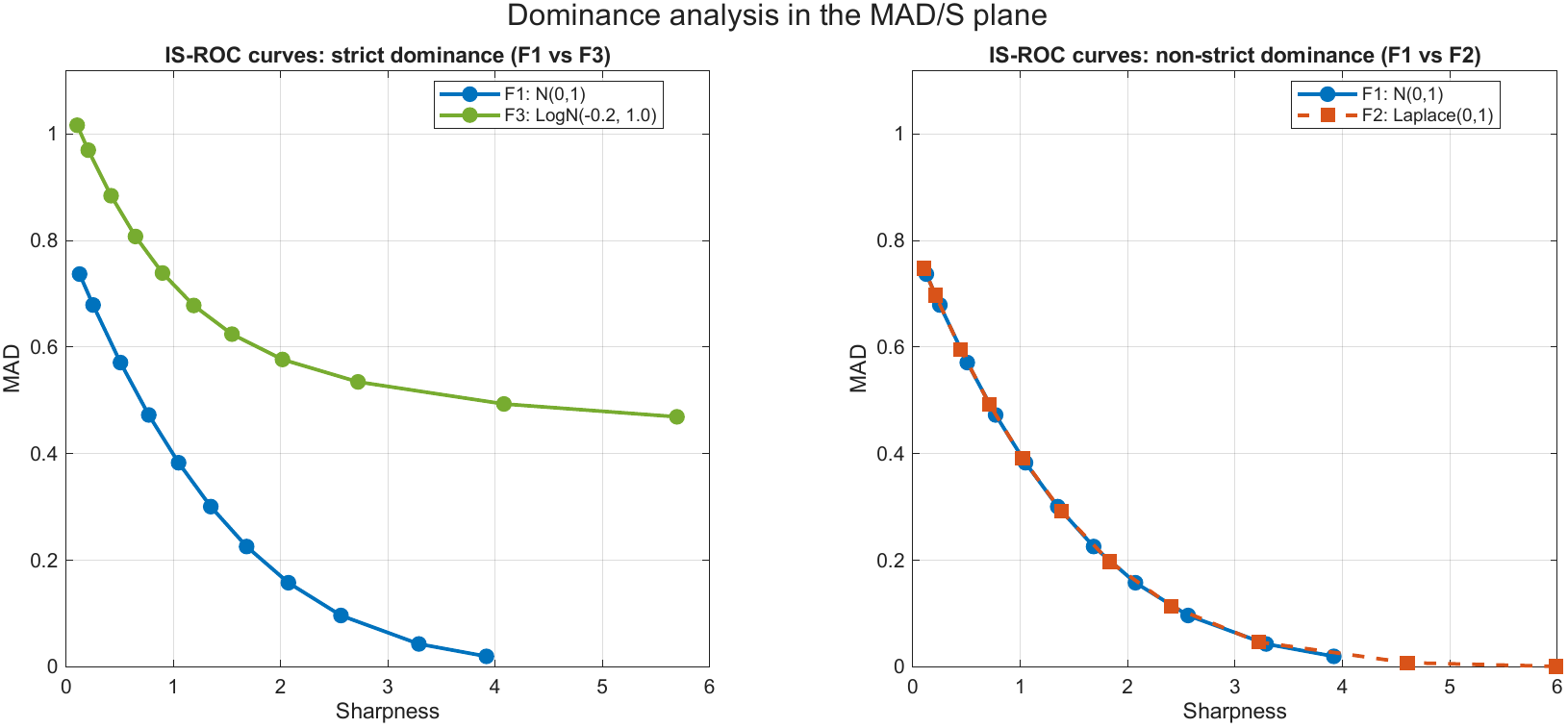}

\caption{
Example 1. Comparison of the IS--ROC Curves.
The left panel compares the Gaussian $(F_1)$ and lognormal $(F_3)$ forecasters, highlighting the dominance of the Gaussian predictor.
The right panel compares the Gaussian $(F_1)$ and Laplace $(F_2)$ forecasters, showing that the two IS--ROC Curves perfectly overlap. As stated in Theorem \ref{thm:nonuniqueness}, this result is a direct consequence of the symmetry of the two predictive distributions and their common median. The difference between the two curves consists of the different parametrization by the tightness parameter. The calibration procedure will directly address this issue.
}

\label{fig:example1_isroc}

\end{figure}

\begin{figure}[H]

\centering

\includegraphics[width=\textwidth]{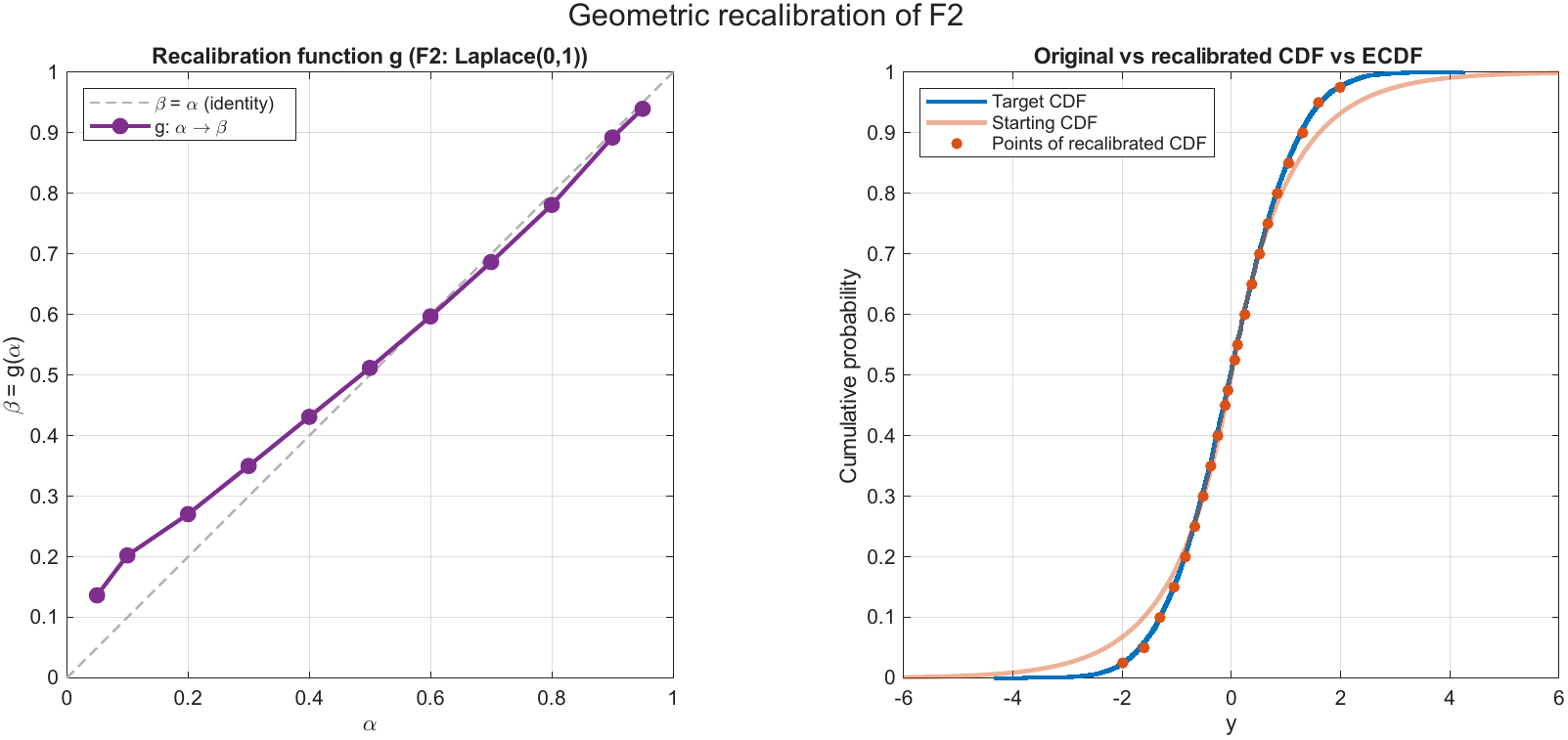}

\caption{
Example 1. The left panel illustrates the calibration function, which maps each desired coverage level $\alpha$ to the corresponding tightness level of the Laplace forecaster $(F_2)$ to be used in practice. The right panel reports the resulting calibrated cumulative distribution functions. As can be seen, the calibrated Laplace predictive distribution (red) closely overlaps the true data-generating distribution (blue).
}

\label{fig:example1_calibration}

\end{figure}

The first comparison illustrates the Pareto dominance of the Gaussian forecaster over the lognormal predictor.
In contrast, the second comparison shows that two different predictive distributions may generate exactly the same IS--ROC Curve.
This confirms that the IS--ROC Curve is not, in general, sufficient to uniquely identify the underlying predictive distribution.
Nevertheless, once the calibration function has been computed, the original predictive distribution can be recovered through Equation (\ref{eq:calibrated_tau}).

\subsection{Example 2: Convex hull and ensembling of IS--ROC Curves}
\label{subsec:example2}

The second example illustrates the situation in which no globally dominant forecaster exists.
We consider a binary covariate
$Z\sim\mathrm{Bernoulli}(0.5)$, and define the DGP as

\[
G=
\begin{cases}
F_1,& Z=0,\\
F_2,& Z=1,
\end{cases}
\]
where
$F_1=\mathrm{Laplace}(0.3,0.6)$,
and
$F_2=\mathcal N(-0.3,2)$. We hence compare three forecasters: $F_1, F_2$ and the oracle (unknown in practical cases). The IS--ROC Curves corresponding to $F_1$ and $F_2$ intersect and no predictor globally dominates the others.
Consequently, the convexification procedure described in Section~\ref{subsec:multiple} is required.

\begin{figure}[H]

\centering

\includegraphics[width=\textwidth]{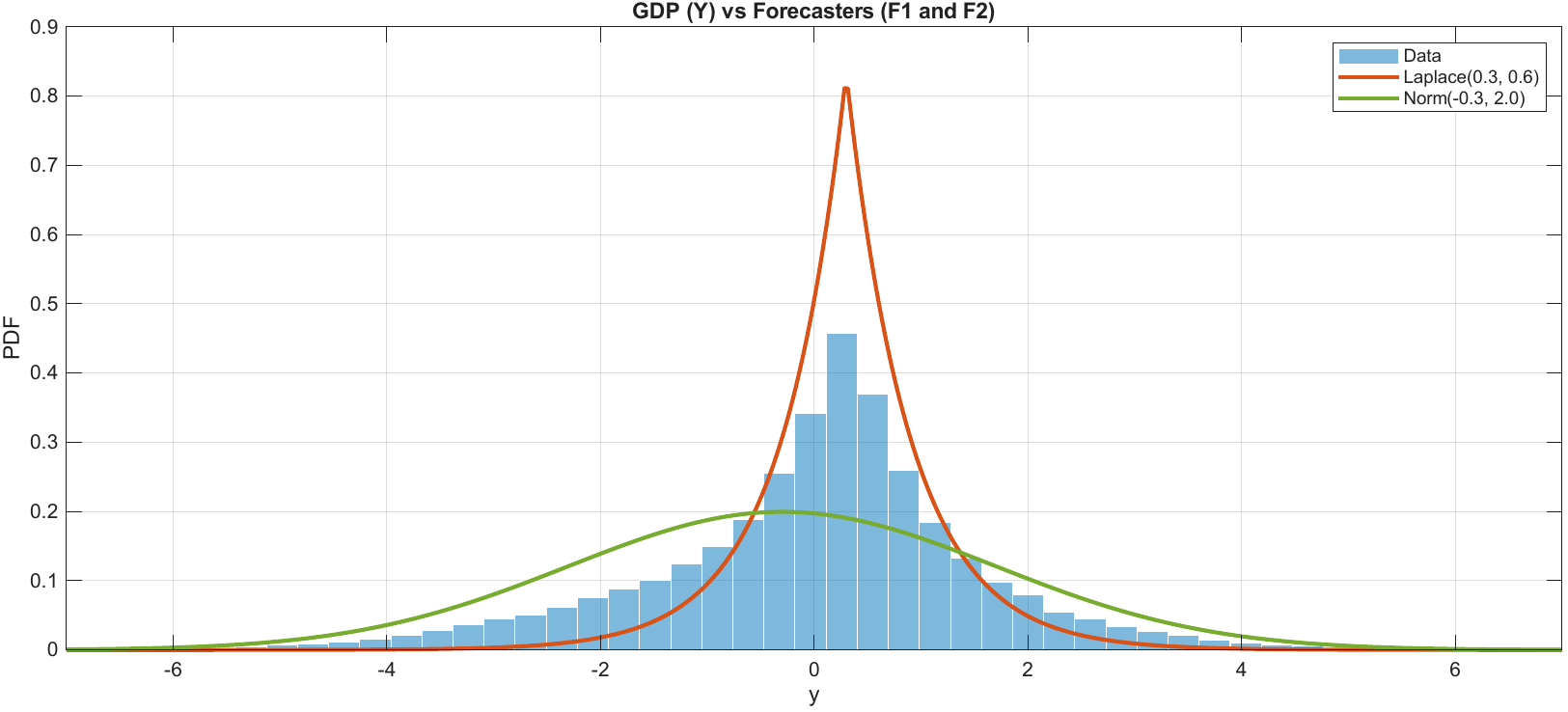}

\caption{
Example 2. Probability density functions corresponding to $F_1$, $F_2$ and the marginal of the oracle, here represented by an histogram.
}

\label{fig:example2_pdf}

\end{figure}

\begin{figure}[H]

\centering

\includegraphics[width=\textwidth]{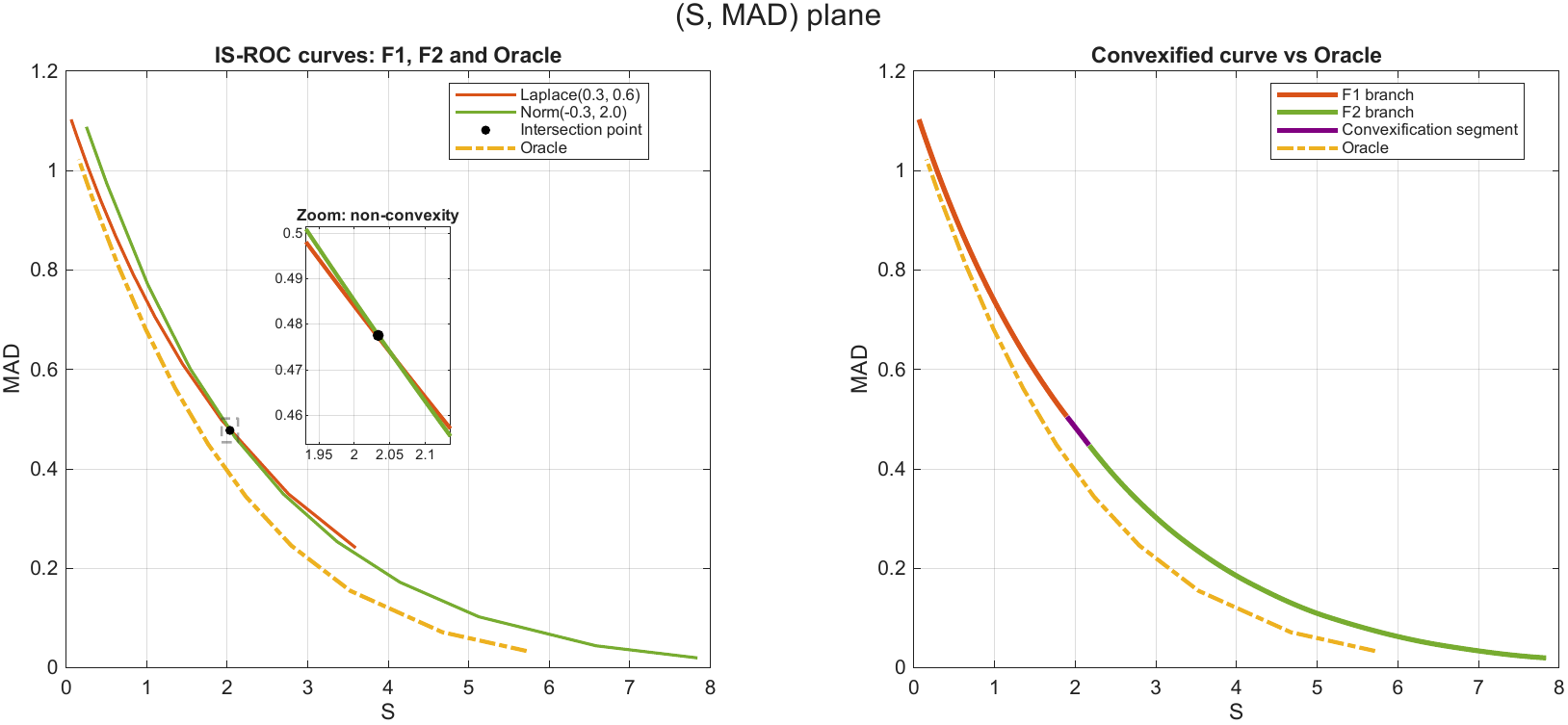}

\caption{Example 2. IS--ROC Comparison. The left panel reports the original IS--ROC Curves: while $F_1$ achieves lower misspecification for small sharpness values, corresponding to prediction intervals centered around the median of the DGP, $F_2$ becomes preferable for larger sharpness values. The intersection of the two curves is marked by the black dot. As a result, the lower envelope is composed of two distinct branches and is globally non-convex. Applying the convexification procedure yields the right panel, where the convexifying segment is shown in purple.
}

\label{fig:example2_convex}

\end{figure}

\begin{figure}[H]

\centering

\includegraphics[width=\textwidth]{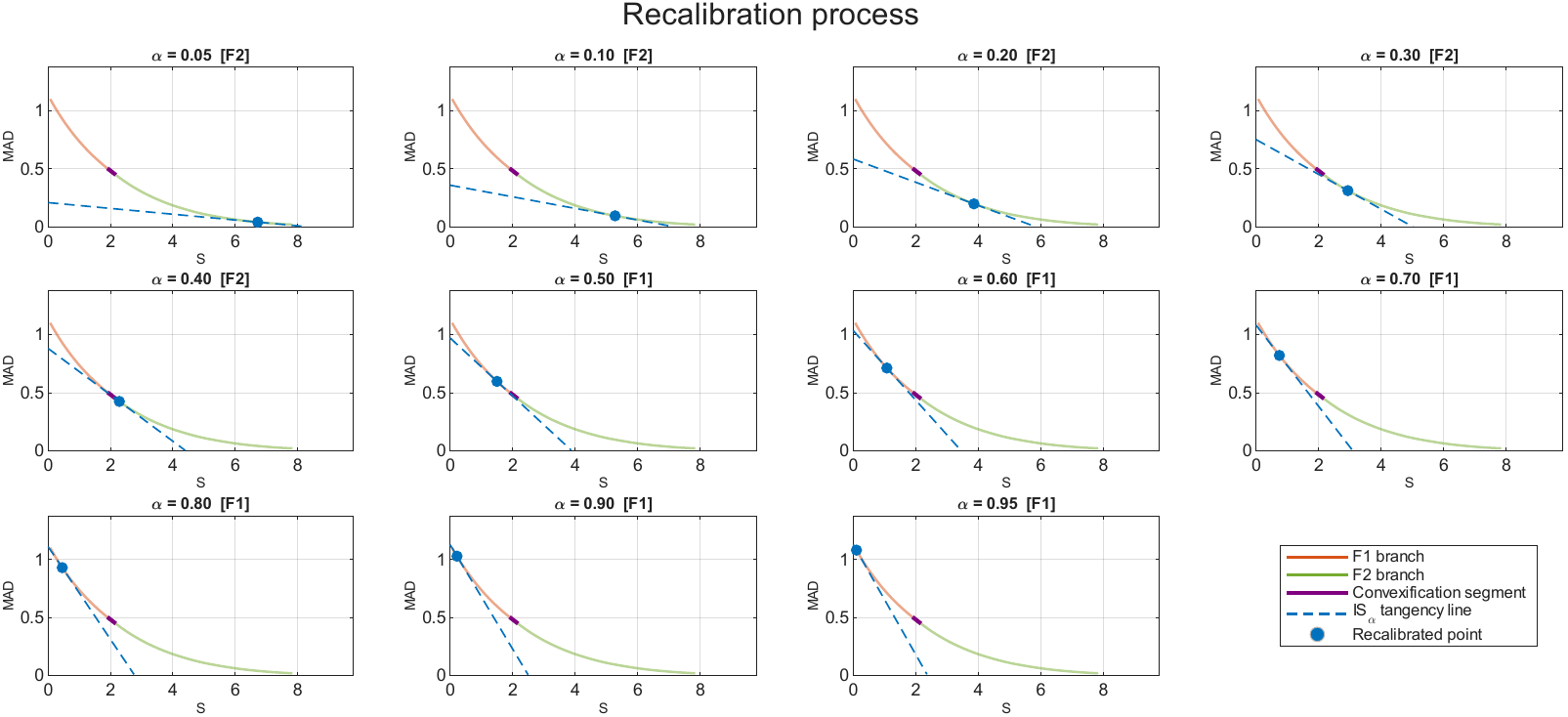}

\caption{Example 2. Calibration process on the MAD--S plane for intersecting IS--ROC curves. Each panel corresponds to a value of the Interval Score parameter $\alpha$ and shows the supporting line of slope $-\alpha/2$ tangent to the global efficient frontier. The original branches inherit their own calibration. For each panel, the marker indicates the calibrated point $\beta = g(\alpha)$ that minimizes $\mathrm{IS}(\alpha,\beta)$ along the frontier, alternating between the $F_1$ and the $F_2$ branches.}

\label{fig:example2_calibration}

\end{figure}

\begin{figure}[H]

\centering

\includegraphics[width=\textwidth]{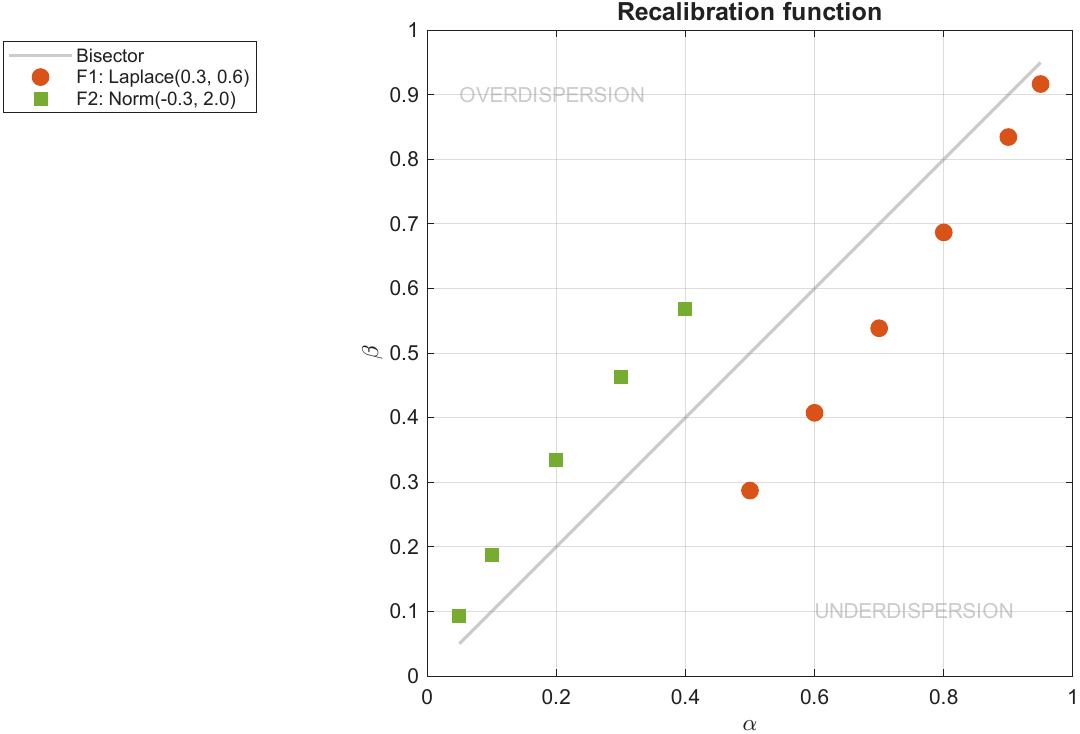}

\caption{
Example 2. Step calibration function associated with the convexified frontier.
Points belonging to different IS--ROC Curves are treated independently.
}

\label{fig:example2_step}

\end{figure}

\begin{figure}[H]

\centering

\includegraphics[width=\textwidth]{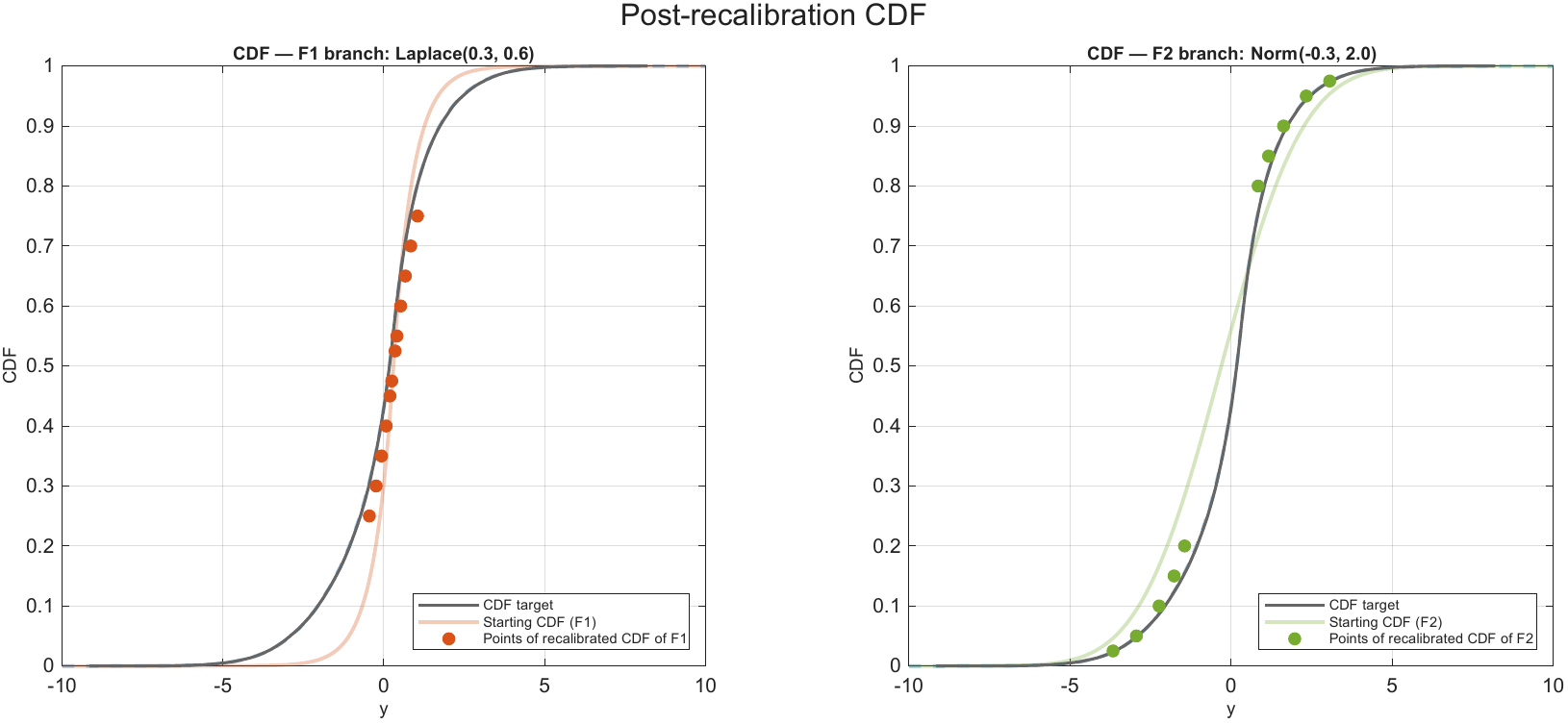}

\caption{
Example 2. Calibrated cumulative distribution functions corresponding to the two conditional cases identified by the covariate.
}

\label{fig:example2_cdf}

\end{figure}

The convexified forecaster improves the frontier of the two competing predictors.
The subsequent calibration procedure preserves this frontier while restoring the optimal correspondence between the Interval Score parameter and the interval tightness.










\section{A roadmap for the design of probabilistic predictors}
The IS--ROC Curve is intended not only as an evaluation tool, but also as a framework for guiding the entire development of probabilistic forecasting models. It naturally supports an iterative workflow in which progressively richer predictors are constructed, assessed, calibrated, and eventually combined.

A natural starting point consists of covariate-free probabilistic predictors, namely TIPs describing only the marginal distribution of the response variable. Their IS--ROC Curves provide a baseline against which all subsequent models can be compared. As additional covariates are introduced, the corresponding IS--ROC Curves immediately reveal whether the increased model complexity translates into a genuine improvement of the sharpness--accuracy trade-off.

Once several candidate predictors have been developed, the methodology proposed in this paper provides a principled strategy for their refinement. Each predictor can first be recalibrated through the geometric calibration procedure described in Section~\ref{sec:calibration}. Subsequently, whenever no single predictor globally dominates the others, convexification naturally leads to the construction of an ensemble TIP that combines the most informative operating regions of the competing models.

An interesting extension, which is left for future investigation, is the analysis of \emph{conditional} IS--ROC Curves. Rather than evaluating predictive performance over the entire population, one could study the IS--ROC Curve conditional on one or more covariates. For instance, if a binary covariate such as sex is available, two conditional IS--ROC Curves could be constructed, one for each subgroup. Calibration, convexification, and ensembling could then be performed independently within each conditional probability space before combining the resulting predictors. Such a conditional approach has the potential to further improve predictive performance by adapting the forecasting strategy to different data regimes.


\appendix

\section*{Appendix}


\section{Convexification Lemma}
\label{app:convexification}

This appendix provides the technical result underlying the convexification procedure introduced in Section~\ref{subsec:convexification}.

The lemma shows that every point belonging to a convexification segment can be associated to an interval forecasting strategy. More precisely, every point of the segment corresponds to a randomized interval forecaster and therefore represents an admissible operating point in the \((S,MAD)\)-plane.

\begin{lemma}
\label{lem:convexification}
Consider a data generating process DGP.
Let $P^A$ and $P^B$ be two fixed interval forecasters whose average performance yields the points 
\[
A=(\mathrm{MS}^{A},\MAD^{A}),
\qquad
B=(\mathrm{MS}^{B},\MAD^{B})
\]
in the \((\MS,\MAD)\)-plane.
Then, for every $t\in[0,1]$, there exists an interval forecaster \(P_t\) whose performance point is
\[
(\mathrm{MS}(t),\MAD(t))
=
(1-t)A+tB,
\]
that is,
\[
\mathrm{MS}(t)
=
(1-t)\mathrm{MS}^{A}
+
t\mathrm{MS}^{B},
\]
\[
\MAD(t)
=
(1-t)\MAD^{A}
+
t\MAD^{B}.
\]
\end{lemma}

\begin{proof}

Fix \(t\in[0,1]\), and let

\[
Z\sim\mathrm{Bernoulli}(t).
\]

Define the randomized interval forecaster

\[
P_t=
\begin{cases}
P^{A}, & Z=0,\\[1ex]
P^{B}, & Z=1.
\end{cases}
\]

We first compute the sharpness coordinate.

By the law of total expectation,

\[
\begin{aligned}
\MS(t)
&=
\mathbb{E}\!\left[\mathrm{S}(P_t)\right] \\
&=
\mathbb{E}\!\left[\mathrm{S}(P_t)\mid Z=0\right]\Pr(Z=0)
+
\mathbb{E}\!\left[\mathrm{S}(P_t)\mid Z=1\right]\Pr(Z=1)\\
&=
(1-t)\mathbb{E}\!\left[\mathrm{S}(F_\beta^{A})\right]
+
t\mathbb{E}\!\left[\mathrm{S}(F_\beta^{B})\right]\\
&=
(1-t)\MS^{A}
+
t\MS^{B}.
\end{aligned}
\]

The argument for the Mean Absolute Distance is identical:
\[
\begin{aligned}
\MAD(t)
&=
\mathbb{E}\!\left[\MAD(P_t)\right]\\
&=
\mathbb{E}\!\left[\MAD(P_t)\mid Z=0\right]\Pr(Z=0)
+
\mathbb{E}\!\left[\MAD(P_t)\mid Z=1\right]\Pr(Z=1)\\
&=
(1-t)\mathbb{E}\!\left[\MAD(F_\beta^{A})\right]
+
t\mathbb{E}\!\left[\MAD(F_\beta^{B})\right]\\
&=
(1-t)\MAD^{A}
+
t\MAD^{B}.
\end{aligned}
\]
\end{proof}


\section{Additional Results}
\label{app:additional}

This appendix collects several additional analytical results and examples illustrating the properties of the proposed IS--ROC framework.

\subsection{Piece-wise Linear IS--ROC Curves}
\label{app:linear}

An interesting special case is that of a piecewise linear IS--ROC Curve, which may arise for discrete DGP.
Consider the random variable

\[
Y=
\begin{cases}
-4, & \text{with probability }0.30,\\
-2, & \text{with probability }0.10,\\
0,  & \text{with probability }0.40,\\
2,  & \text{with probability }0.20.
\end{cases}
\]

Figure~\ref{fig:linear} reports the corresponding probability distribution together with the resulting IS--ROC Curve.

\begin{figure}[H]
    \centering
    \includegraphics[width=\textwidth]{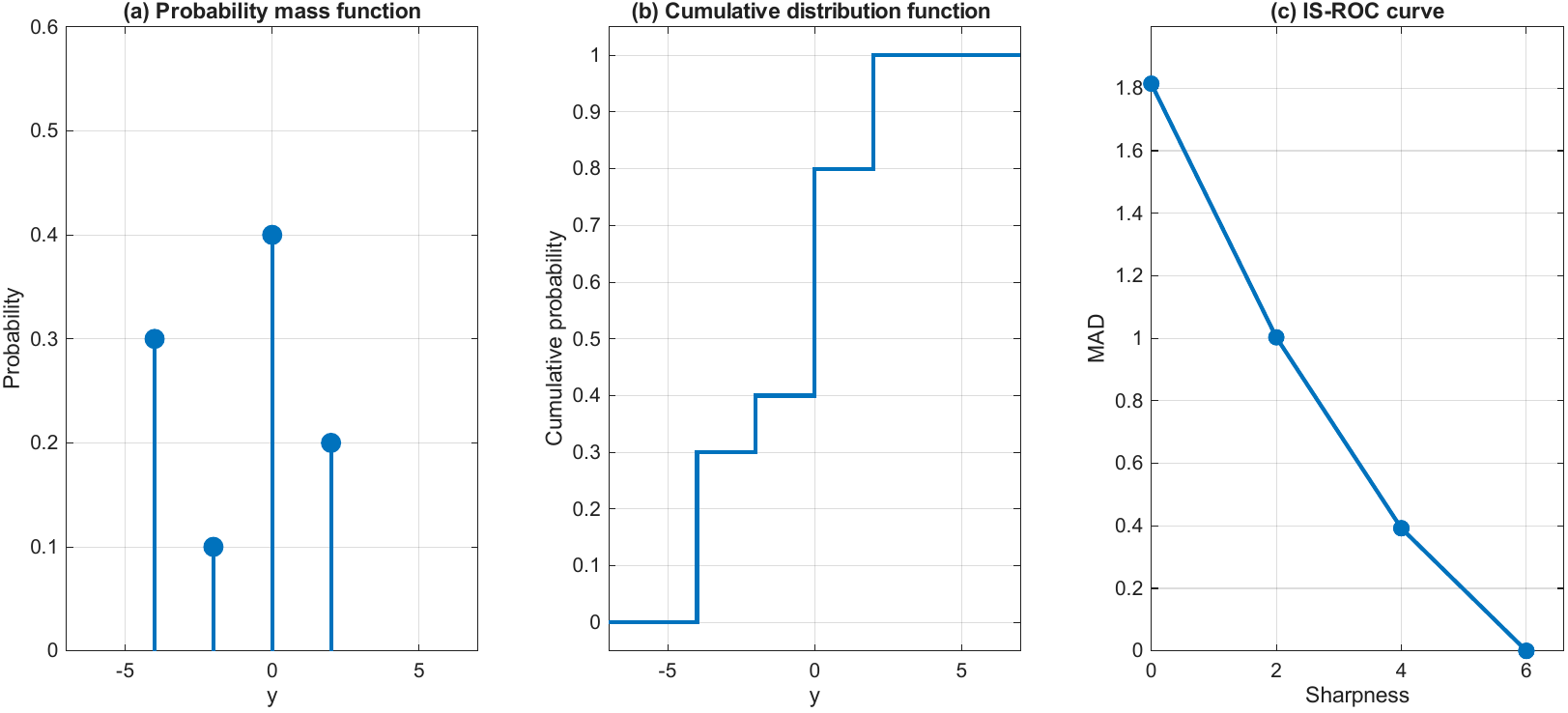}
    \caption{
    Example of a linear IS--ROC Curve.
    The left panel shows the probability distribution, whereas the right panel reports the corresponding IS--ROC Curve.
    }
    \label{fig:linear}
\end{figure}



\subsection{Quasi-Convexity of the IS--ROC Curves}
\label{app:monotonicity}

By assumption, for every TIP $P$, the sharpness is a monotone function of the tightness parameter.
More precisely,
\[
\beta_1<\beta_2
\quad\Longrightarrow\quad
\mathrm{S}(\beta_1)\ge\mathrm{S}(\beta_2).
\]
Conversely, enlarging the prediction interval can only decrease the Absolute Distance:
\[
\beta_1<\beta_2
\quad\Longrightarrow\quad
\MAD(\beta_1)\le\MAD(\beta_2).
\]
Since averaging preserves monotonicity, we have that any IS--ROC Curve is monotone nonincreasing in the \((S,MAD)\)-plane.
In turn, the monotonicity of $\mathrm{S}\mapsto\MAD$ immediately implies that every IS–-ROC Curve is quasi-convex.

An illustrative example is obtained by considering a synthetic power generation process modeled as a trimodal mixture of Beta distributions. The three components represent low-, medium-, and high-wind regimes, with mixing probabilities $0.30$, $0.40$, and $0.30$, respectively. The corresponding Beta distributions are $Beta(1.2,22.8)$, $Beta(12,12)$, and $Beta(22.8,1.2)$, yielding three symmetric regimes centered around low, medium, and high power production. The true predictor correctly identifies the generating regime and therefore issues forecasts based on the corresponding Beta distribution. In contrast, the misspecified predictor correctly recognizes the two extreme regimes but cannot distinguish the intermediate one. Whenever the medium-wind regime occurs, it instead predicts using the unconditional mixture of the three Beta components, thereby increasing predictive uncertainty while preserving calibration at the population level.

\begin{figure}[H]
    \centering
    \includegraphics[width=0.7\textwidth]{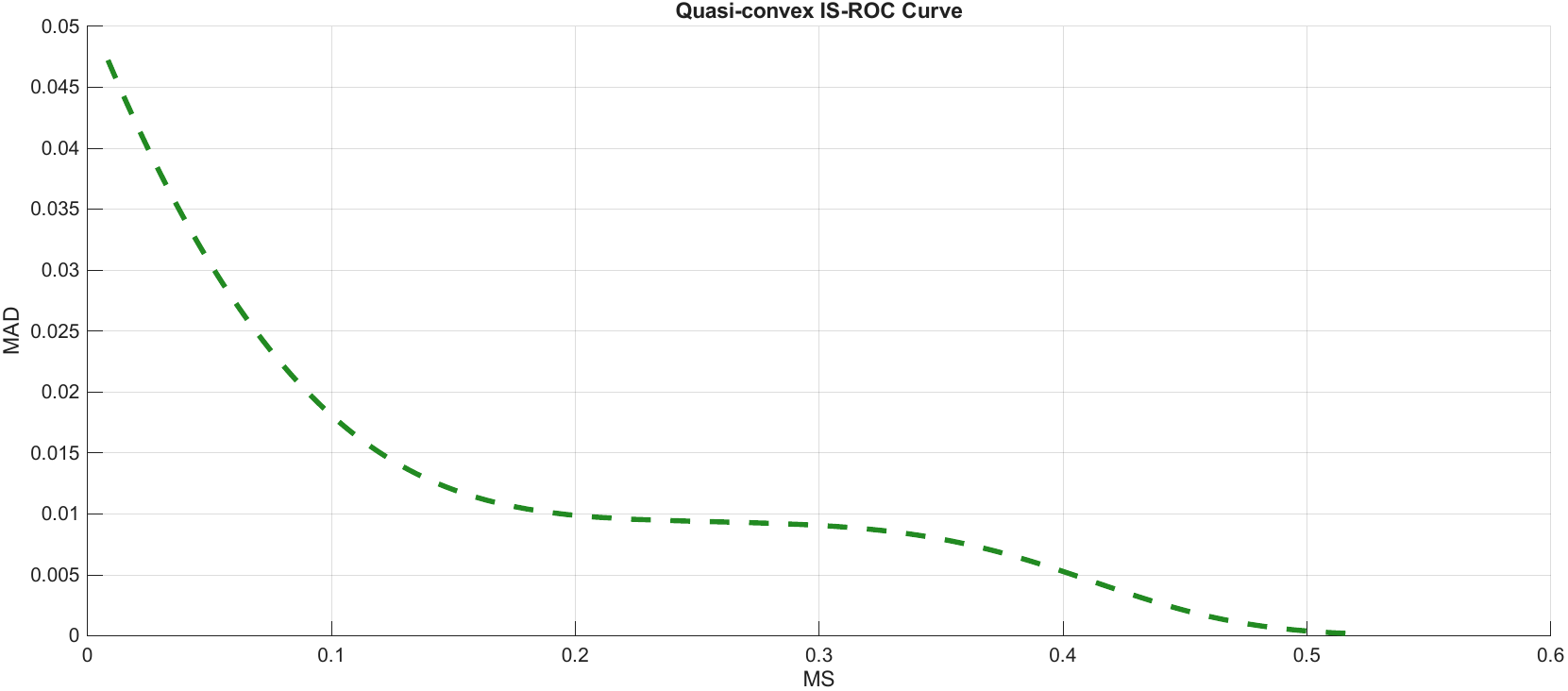}
    \caption{
     The quasi-convex IS–ROC Curve associated to the misspecified predictors is shown.
    }
    \label{fig:quasiconvex}
\end{figure}

\subsection{Sharpness curve does not identify the distribution}
\label{app:sharpness}

Sharpness alone is generally insufficient to identify a predictive distribution.

Indeed, different predictive distributions may generate exactly the same mean sharpness function
$\beta
\longmapsto
\MS(\beta),
$
while producing different Mean Absolute Distance functions.

An immediate example is provided by Gaussian distributions having identical variance but different medians. Since the interval width depends only on the variance, all such distributions generate the same sharpness curve.

A more interesting question concerns if two distinct distributions with the same median can generate the same sharpness diagram.
Figure~\ref{fig:sharpness} reports an illustrative example.

\begin{figure}[H]
    \centering
    \includegraphics[width=0.8\textwidth]{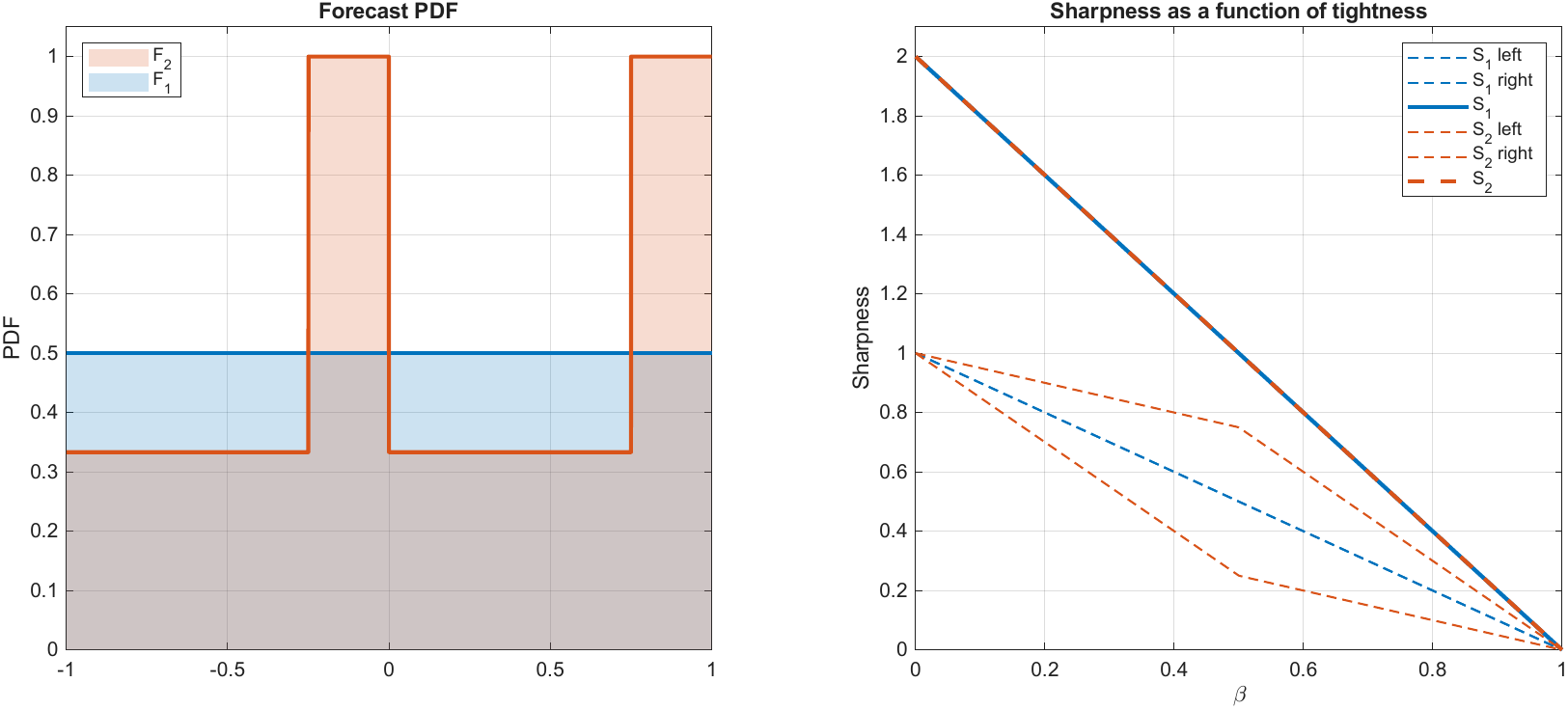}
    \caption{
    Two predictive distributions with the same median producing identical sharpness diagrams. The left panel shows the predictive densities $F1$ (blue) and $F2$ (orange). The right panel reports the corresponding sharpness functions as a function of the tightness parameter $\beta$. Solid lines denote the overall sharpness, which is identical for the two forecasts. Dashed lines represent the contributions of the left and right semi-widths, revealing that the two predictive distributions allocate uncertainty differently on the two sides of the median despite having the same total sharpness.
    }
    \label{fig:sharpness}
\end{figure}


\subsection{Necessary condition to produce the same IS--ROC Curve}
\label{app:median}

Although different predictive distributions may generate the same IS--ROC Curve, a necessary condition is that they share the same median.
Indeed, the prediction intervals considered throughout this work are centered around the predictive median.
Consequently, if two IS--ROC Curves coincide for every value of the tightness parameter, then $\mathrm{S}(1)=0$ corresponds to the same prediction point, implying coincidence of the corresponding medians.

\subsection{Closed-form expressions}
\label{app:closedforms}

While one can always compute
$\mathrm{S}(\beta)$ and $\MAD(\beta),$
obtaining an explicit parametric representation of the IS--ROC Curve, the direct representation $\mathrm{S}\mapsto \MAD(s)$ is trickier to find.
Whenever the mapping $\beta\longmapsto S(\beta)$
is invertible, one also obtains the closed-form relation between sharpness and $\MAD$.
The following tables summarize the analytical expressions for the distributions considered in this work.

\begin{landscape}

\begin{table}[H]
\centering
\setlength{\tabcolsep}{6pt}

\begin{tabular}{lcc}
\toprule
\textbf{Distribution} 
& \textbf{Sharpness } $s(\beta)$ 
& $\mathbf{MAD}_Y^F(\beta)$ \\
\midrule

Normal $\mathcal{N}(m,\sigma^2)$ 
&
$2\sigma\Phi^{-1}(1-\beta/2)$
&
$(m+\sigma z_{\beta/2}-Y)_+
+(Y-m-\sigma z_{\beta/2})_+$ 
\\[0.3cm]

Lognormal $\mathrm{LN}(m,\sigma)$ 
&
$e^m(e^{-\sigma z}-e^{\sigma z})$, 
$z=\Phi^{-1}(\beta/2)$
&
$(e^{m+\sigma z}-Y)_+
+(Y-e^{m-\sigma z})_+$
\\[0.3cm]

Exponential $\mathrm{Exp}(\lambda)$
&
$\frac{1}{\lambda}
\log\frac{1-\beta/2}{\beta/2}$
&
$\left(-\frac1\lambda\log(1-\beta/2)-Y\right)_+
+
\left(Y+\frac1\lambda\log(\beta/2)\right)_+
$
\\[0.3cm]

Uniform $\mathrm{Unif}(a,b)$
&
$(1-\beta)(b-a)$
&
$\left(a+\frac{\beta}{2}(b-a)-Y\right)_+
+
\left(Y-a-(1-\frac{\beta}{2})(b-a)\right)_+
$

\\

\bottomrule
\end{tabular}

\caption{Sharpness and MAD expressions in function of tightness for selected forecast distributions.}
\label{tab:sharpness_mad}

\end{table}

\begin{table}[H]
\centering
\small
\setlength{\tabcolsep}{6pt}

\begin{tabular}{lcc}
\toprule
\textbf{Distribution}
& \textbf{Inverse } $\beta(s)$
& $\mathbf{MAD}_Y^F(s)$
\\
\midrule

Normal $\mathcal{N}(m,\sigma^2)$
&
$2\Phi(-s/(2\sigma))$
&
$\left(m-\frac{s}{2}-Y\right)_+
+
\left(Y-m-\frac{s}{2}\right)_+
$
\\[0.3cm]

Lognormal $\mathrm{LN}(m,\sigma)$
&
$2\Phi
\left(
\frac1\sigma
\operatorname{arsinh}
(-\frac{s}{2e^m})
\right)$
&
$(l(s)-Y)_+
+(Y-u(s))_+$
\\[0.2cm]

&
&
$l(s)=e^m
\left(
-\frac{s}{2e^m}
+\sqrt{1+(\frac{s}{2e^m})^2}
\right)$
\\

&
&
$u(s)=e^m
\left(
\frac{s}{2e^m}
+\sqrt{1+(\frac{s}{2e^m})^2}
\right)$
\\[0.3cm]

Exponential $\mathrm{Exp}(\lambda)$
&
$\frac{2}{1+e^{\lambda s}}$
&
$\left(
-s+\frac1\lambda\log(1+e^{\lambda s})-Y
\right)_+
+
\left(
Y-\frac1\lambda\log(1+e^{\lambda s})
\right)_+
$
\\[0.3cm]

Uniform $\mathrm{Unif}(a,b)$
&
$1-\frac{s}{b-a}$
&
$\left(
\frac{a+b}{2}-\frac{s}{2}-Y
\right)_+
+
\left(
Y-\frac{a+b}{2}-\frac{s}{2}
\right)_+
$

\\

\bottomrule
\end{tabular}

\caption{Tightness in function of Sharpness, and MAD rin function of Sharpness for selected forecast distributions.}
\label{tab:inverse_sharpness_mad}

\end{table}

\end{landscape}





\bibliographystyle{unsrt}
\bibliography{references}

\end{document}